\documentclass[letterpaper, 10 pt, conference]{ieeeconf}  

\IEEEoverridecommandlockouts                              
\usepackage{amsmath} 
\usepackage{amssymb}  

\usepackage{comment}
\usepackage{amsthm} 
\theoremstyle{plain}
\newtheorem{thm}{Theorem}
\newtheorem{lem}{Lemma}
\newtheorem{cor}{Corollary}

\usepackage[scr]{rsfso}
\usepackage{mathdots}
\usepackage{empheq} 
\usepackage[nolist]{acronym}    

\newcommand{\xl}{\mathbf{x}}
\newcommand{\ul}{\mathbf{u}}
\newcommand{\yl}{\mathbf{y}}
\newcommand{\el}{\mathbf{e}}
\newcommand{\dl}{\mathbf{d}}
\newcommand{\G}[2]{\mathcal{G}_{#1}^{#2}}
\newcommand{\uG}[2]{{}^\mathrm{u}\mathcal{G}_{#1}^{#2}}
\newcommand{\dG}[2]{{}^\mathrm{d}\mathcal{G}_{#1}^{#2}}
\newcommand{\eG}[2]{{}^\mathrm{e}\mathcal{G}_{#1}^{#2}}
\newcommand{\blkH}[4]{\mathcal{H}_{#1,#2,#3}(#4)}
\newcommand{\cvec}[3]{{#1}_{[#2,\;#3]}}
\newcommand{\xix}{\xi_{\mathrm{x}}}
\newcommand{\xid}{\xi_{\mathrm{d}}}
\newcommand{\nxd}{{\bar{n}}}
\newcommand{\nd}{{n_\mathrm{d}}}
\newcommand{\nx}{{n_\mathrm{x}}}
\newcommand{\xbar}{\bar{x}}
\newcommand{\Al}{\mathcal{A}_\mathrm{L}}
\newcommand{\Bl}{\mathcal{B}_\mathrm{L}}

\newcommand{\thmref}[1]{Theorem~\ref{#1}}
\newcommand{\lemref}[1]{Lemma~\ref{#1}}
\newcommand{\corref}[1]{Corollary~\ref{#1}}

\title{\LARGE \bf
Data-enabled Predictive Repetitive Control
}

\author{Rogier Dinkla$^{1}$, Tom Oomen$^{1,2}$, Sebastiaan P. Mulders$^{1}$ and Jan-Willem van Wingerden$^{1}$
\thanks{$^{1}$Delft University of Technology, Delft Center for Systems and Control, Mekelweg 2, 2628 CD Delft, The Netherlands.
        {\tt\small \{R.T.O.Dinkla, T.A.E.Oomen, S.P.Mulders, J.W.vanWingerden\}@tudelft.nl}.}
\thanks{$^{2}$Eindhoven University of Technology, Control Systems Technology Group, 5600 MB Eindhoven, The Netherlands.}
}

\begin{document}

\maketitle
\thispagestyle{empty}
\pagestyle{empty}

\begin{abstract}
Many systems are subject to periodic disturbances and exhibit repetitive behaviour. Model-based repetitive control employs knowledge of such periodicity to attenuate periodic disturbances and has seen a wide range of successful industrial implementations. The aim of this paper is to develop a data-driven repetitive control method. In the developed framework, linear periodically time-varying (LPTV) behaviour is lifted to linear time-invariant (LTI) behaviour. Periodic disturbance mitigation is enabled by developing an extension of Willems' fundamental lemma for systems with exogenous disturbances. The resulting Data-enabled Predictive Repetitive Control (DeePRC) technique accounts for periodic system behaviour to perform attenuation of a periodic disturbance. Simulations demonstrate the ability of DeePRC to effectively mitigate periodic disturbances in the presence of noise.
\end{abstract}
\begin{acronym}%
    \acro{LTI}{linear time-invariant}
    \acro{LPTV}{linear periodically time-varying}
    \acro{SPRC}{Subspace Predictive Repetitive Control}
    \acro{DeePC}{Data-enabled Predictive Control}
    \acro{DeePRC}{Data-enabled Predictive Repetitive Control}
    \acro{VRFT}{Virtual Reference Feedback Tuning}
    \acro{IFT}{Iterative Feedback Tuning}
    \acro{CbT}{Correlation-based Tuning}
    \acro{MFAC}{Model-free Adaptive Control}
    \acro{ILC}{Iterative Learning Control}
    \acro{RC}{Repetitive Control}
    \acro{FRF}{Frequency Response Function}
    \acro{MIMO}{multiple-input multiple-output}
    \acro{CL-DeePC}{Closed-loop Data-enabled Predictive Control}
\end{acronym}%
\section{INTRODUCTION}
Periodic disturbances and repetitive behaviour are encountered in many systems such as wind turbines~\cite{vanVondelen2023} and semiconductor manufacturing~\cite{vanZundert2017}. The periodic nature of disturbances is exploited by repetitive control to obtain improved reference tracking performance compared to conventional feedback control. Model-based repetitive control uses the internal model principle~\cite{Francis1976} to model a periodic disturbance by means of a memory loop, thereby facilitating complete attenuation of errors that share the same periodicity as the disturbance~\cite{Hara1988}.

The pursuit of fast and accurate repetitive controllers has prompted many model-based forms of repetitive control. Modelling has allowed the design of potentially non-causal filters to enhance robustness and learning~\cite{deRozario2019}. In the frequency domain, \ac{FRF} data enables uncertainty modelling for robust controller designs~\cite{DeRoover2000}. Repetitive control of \ac{MIMO} systems has been facilitated by means of, for example,  $\mathcal{H}_\infty$ techniques~\cite{Amann1996}. Unfortunately, repetitive control applications often rely on parametric models that are costly and hard to obtain due to the complexity that arises with, e.g., underdamped mechanical systems~\cite{Gevers2005}.

The combination of widespread availability of data and increasing system complexity motivates direct data-driven repetitive control designs. Unlike model-based control, direct data-driven techniques do not rely on an intermediate synthesis of a parametric model~\cite{Markovsky2022}, thereby alleviating the need for the aforementioned costly modelling of systems with complex dynamics~\cite{Hou2013}. Recently, the development of the direct data-driven predictive control technique \ac{DeePC}~\cite{Coulson2019} has garnered significant interest. \ac{DeePC} applies Willems' fundamental lemma~\cite{Willems2005} in a receding horizon framework. An important feature of \ac{DeePC} is its innate ability to handle constraints due to its reliance on optimization to solve an optimal control problem.

Despite attracting considerable attention, typical \ac{DeePC} approaches do not account for periodic system behaviour and cannot incorporate periodic exogenous disturbance generators. \ac{DeePC} has recently been incorporated in \ac{ILC}~\cite{Chu2023,Zhang2024}, in which field periodic disturbances feature prominently, but by the nature of \ac{ILC} these applications are limited to cases where the system's state resets periodically. In addition, to mitigate a periodic disturbance, standard \ac{DeePC} based on Willems' fundamental lemma is insufficient because it assumes system controllability. Whilst \ac{DeePC} has been extended to linear parameter-varying~\cite{Verhoek2021} and \ac{LPTV}~\cite{Li2022} systems, periodic disturbance attenuation for repetitive control applications is not considered.

Although \ac{DeePC} has seen considerable development, its use to attenuate periodic disturbances and accommodate periodic dynamics is not adequately addressed. This paper's main contribution is the development of a \ac{DeePC}-inspired repetitive control framework named \ac{DeePRC} that attenuates the influence of periodic disturbances and accommodates \ac{LPTV} system behaviour. Building on~\cite{Li2022} a technique called `lifting' is used to transform \ac{LPTV} to \ac{LTI} dynamics. A suitable relaxation of the system controllability assumed by Willems' fundamental lemma is developed to faciliate the mitigation of periodic disturbances. Furthermore, by incorporating \ac{CL-DeePC}~\cite{Dinkla2024} in a lifted framework, \ac{DeePRC} relies on a computationally efficient implementation that can adequately mitigate noise, including during closed-loop operation.

This paper is organized as follows. Section~\ref{sec:prelim} introduces the employed \ac{LPTV} model, notation and definitions. The \ac{DeePRC} framework is developed in Section~\ref{sec:DeePRC}. To this end, the \ac{LPTV} system is first lifted to an \ac{LTI} representation in Section~\ref{sec:lifting}. Then, in Section~\ref{sec:IMP}, the internal model principle is used to motivate augmenting the lifted state with a constant disturbance. Section~\ref{sec:DistWFL} explains how the controllability assumption of Willems' fundamental lemma may be relaxed to accomodate such disturbances in a \ac{DeePC} framework, which motivates the \ac{DeePRC} formulation presented in Section~\ref{sec:DeePRC_formulation}. Thereafter, a simulation case study is presented in Section~\ref{sec:Simulation}, and conclusions and suggestions for future work are provided in Section~\ref{sec:Conclusion}.
\section{PRELIMINARIES}\label{sec:prelim}
\subsection{Periodic System Model}
This paper considers the signal generating plant to be a discrete-time \ac{LPTV} system $\mathcal{S}$ in innovation form to capture the effects of process and measurement noise~\cite{Verhaegen2007a}
\begin{subequations}\label{eq:SS_inno}
\begin{empheq}[left=\mathcal{S}:\empheqlbrace]{align}
    x_{k+1} &= A_k x_k + B_k u_k + F_k d_k + K_k e_k\label{eq:SSinno_x}\\
	y_k &= C_k x_k + D_k u_k + G_k d_k + e_k \label{eq:SSinno_y}
  \end{empheq}
\end{subequations}
where the subscript $k\in\mathbb{Z}$ is used as a discrete time index, $\{A_k,B_k,F_k,K_k,C_k,D_k,G_k\}$ are periodic system matrices, ${x_k\in\mathbb{R}^n}$ represents the system's states, ${u_k\in\mathbb{R}^r}$ are inputs, ${d_k\in\mathbb{R}^m}$ are periodic disturbances, ${y_k\in\mathbb{R}^l}$ are outputs, and $e_k\in\mathbb{R}^l$ is zero mean white innovation noise. The disturbances $d_k$ and the system $\mathcal{S}$ are assumed to be $P$-periodic. For disturbances this entails that ${d_k=d_{k+P}}$ whilst for the system $\mathcal{S}$ the definition below applies.\\
\textbf{Definition 1.\label{def:P-periodic} ($\mathbf{P}$-periodic \ac{LPTV} system)~\cite{Bamieh1991}:} a causal system $\mathcal{S}$ is said to be $P$-periodic if it commutes with the delay operator $\mathscr{D}_P$ such that $\mathscr{D}_P \mathcal{S}=\mathcal{S} \mathscr{D}_P$, where $(\mathscr{D}_P f)(k):=f(k-P)$ for a function of time $f$.

In effect, for the state space representation of \eqref{eq:SS_inno}, $P$-periodicity entails that for all of the system matrices ${A_k=A_{k+P}}$, ${B_k=B_{k+P}}$, etc.

\subsection{Notation and Definitions}
This section introduces notation and definitions that are used throughout this paper. For discrete time indices ${k_1,k_2\in\mathbb{Z}}$ with $k_2\geq k_1$ we start by defining a monodromy matrix $\Phi_{k_1}^{k_2}$ and Markov parameters $\mathcal{G}_{k_1}^{k_2}$:
\begin{align}
\Phi_{k_1}^{k_2} &:=
\begin{cases}
        A_{k_2-1}A_{k_2-2}\cdots A_{k_1}, &k_2>k_1 \\
        I_n, &k_2=k_1,
\end{cases}\label{eq:monodromy}\\
\G{k_1}{k_2}(\mathfrak{B}_k,\mathfrak{D}_k) &:= 
\begin{cases}
        \mathfrak{D}_{k_1}, &\mskip47mu k_2=k_1 \\
        C_{k_2}\Phi_{k_1+1}^{k_2}\mathfrak{B}_{k_1}, &\mskip47mu k_2>k_1,
\end{cases}\label{eq:PeriodicMarkov}
\end{align}
where $I_n\in\mathbb{R}^{n\times n}$ represents an identity matrix, and $\mathfrak{D}_k$ and $\mathfrak{B}_k$ represent (periodic) matrices of the system as in \eqref{eq:SS_inno}. Specific types of Markov parameters are defined by
\begin{align}
    \begin{split}\uG{k_1}{k_2}&:=\G{k_1}{k_2}(B_k,D_k)\quad  
    \dG{k_1}{k_2}:=\G{k_1}{k_2}(F_k,G_k) \\ 
    \eG{k_1}{k_2}&:=\G{k_1}{k_2}(K_k,I_l)
    \end{split}
\end{align}
where $I_l$ is not a periodic matrix but an identity matrix.
Using \eqref{eq:monodromy} and \eqref{eq:PeriodicMarkov}, furthermore define matrices with Markov parameters $\mathcal{T}_{k_1}^{k_2}(\mathcal{G})$, reversed extended controllability matrices $\mathscr{C}_{k_1}^{k_2}(\mathfrak{B}_k)$, and an extended observability matrix $\mathcal{O}_{k_1}^{k_2}$ as
\begin{align}
    \!\!\mathcal{T}_{k_1}^{k_2}(\mathcal{G}) &:= \!
    \begin{bmatrix}
		\mathcal{G}_{k_1}^{k_1}   & 0  & \cdots & 0\\
		\mathcal{G}_{k_1}^{k_1+1} & \mathcal{G}_{k_1+1}^{k_1+1} & \cdots & 0\\
		\vdots & \vdots & \ddots & \vdots\\
		\mathcal{G}_{k_1}^{k_2} & \mathcal{G}_{k_1+1}^{k_2} & \cdots  & \mathcal{G}_{k_2}^{k_2}
	\end{bmatrix},\label{eq:Periodic_FIR}\\
    \!\!\mathscr{C}_{k_1}^{k_2}(\mathfrak{B}_k) &:= \!\left[
        \Phi_{k_1+1}^{k_2+1}\mathfrak{B}_{k_1}\; \Phi_{k_1+2}^{k_2+1}\mathfrak{B}_{k_1+1}\, \cdots\, \Phi_{k_2+1}^{k_2+1}\mathfrak{B}_{k_2}
    \right]\!,\label{eq:Periodic_Rev_Ctrb}\\
    \!\!\mathcal{O}_{k_1}^{k_2} &:= \!\left[
        C_{k_1}\Phi_{k_1}^{k_1};\; C_{k_1+1}\Phi_{k_1}^{k_1+1};\; \cdots;\; C_{k_2}\Phi_{k_1}^{k_2}
    \right]\label{eq:Periodic_Obsv}.
\end{align}

Block-Hankel data matrices are defined as
\begin{align}
    \blkH{i}{s}{N}{u_k}
    = \begin{bmatrix}
    u_i & u_{i+1} & \cdots & u_{i+N-1}\\
    u_{i+1} & u_{i+2} & \cdots & u_{i+N} \\
    \vdots & \vdots & \iddots & \vdots \\
    u_{i+s-1} & u_{i+s} & \cdots & u_{i+N+s-2}
    \end{bmatrix},
\end{align}
where $u_k$ can be replaced with different types of data, $i\in\mathbb{Z}$ indicates the start of the used sequence, and $N,s\in\mathbb{Z}$ respectively indicates the number of columns and block rows of the matrix. The notion of persistency of excitation is defined using block-Hankel matrices as follows.\\
\textbf{Definiton 2. (Persistency of excitation)~\cite{Willems2005}:} The signal given by the sequence $\{w_k\}_{k=i}^{i+N+s-2}$ is called persistently exciting of order $s$ if its block-Hankel matrix $\blkH{i}{s}{N}{w_k}$ is full row rank.

Furthermore, column vectors of concatenated data samples are exemplified by
\begin{align}
    \cvec{u}{k_1}{k_2} := \begin{bmatrix}u_{k_1}^\top & u_{k_1+1}^\top & \cdots & u_{k_2}^\top\end{bmatrix}^\top, \quad k_2\geq k_1.
\end{align}
\section{\acl{DeePRC}}\label{sec:DeePRC}
This section presents the development of \acf{DeePRC}. The periodic system model is first transformed to an \ac{LTI} system by means of lifting. To mitigate periodic disturbances, it is shown that the controllability assumption that underpins Willems' fundamental lemma may be relaxed, thereby facilitating data-driven use of the internal model principle.

\subsection{Lifting the \ac{LPTV} to an \ac{LTI} system}\label{sec:lifting}
Lifting is a technique that is used to transform \ac{LPTV} systems into a higher-dimensional \ac{LTI} representation. This section lifts the $P$-periodic system from \eqref{eq:SS_inno} along the lines of~\cite{Li2022}, from which we obtain the following definition.\\
\textbf{Definition 3.\label{def:lifting} (Lifting):} For a $P$-periodic \ac{LPTV} system $\mathcal{S}$ as in \eqref{eq:SS_inno} the corresponding \ac{LTI}, lifted system representation $\mathcal{S}_\mathsf{L}(k_0)$ of $\mathcal{S}$ with initial time $k_0\in\mathbb{Z}$ is
\begin{subequations}\label{eq:lSS_inno}
\begin{empheq}[left=\mathcal{S}_\mathsf{L}(k_0):\empheqlbrace]{align}
    \xl_{j+1} &= A \xl_j + B \ul_j + F \dl_j + K \el_j\label{eq:lSSinno_x}\\
	\yl_j &= C \xl_j + D \ul_j + G \dl_j + H \el_j \label{eq:lSSinno_y}
  \end{empheq}
\end{subequations}
with iteration index $j\in\mathbb{Z}$, and state $\xl_j\in\mathbb{R}^n$, inputs $\ul_j\in\mathbb{R}^{rP}$, outputs $\yl_j\in\mathbb{R}^{lP}$, disturbance $\dl_j\in\mathbb{R}^{mP}$, and innovation noise $\el_j\in\mathbb{R}^{lP}$. The following relations exist between these quantities in the $P$-periodic \ac{LPTV} system $\mathcal{S}$ and the lifted \ac{LTI} system $\mathcal{S}_\mathsf{L}(k_0)$
\begin{align*}
\xl_j := x_{k_0+jP},
\qquad
\ul_j := \cvec{u}{k_0+jP}{k_0+(j+1)P-1}
\end{align*}
with $\yl_j$, $\dl_j$, $\el_j$ defined akin to $\ul_j$. For the lifted system matrices $\{A,B,F,K,C,D,G,H\}$ in \eqref{eq:lSS_inno} note the lack of a subscript to distinguish them from \ac{LPTV} counterparts. These system matrices are defined using \eqref{eq:monodromy} to \eqref{eq:Periodic_Obsv} as follows:
\begin{align*}
    A&:=\Phi_{k_0}^{k_0+P} &
    C&:=\mathcal{O}_{k_0}^{k_0+P-1} \\
    B&:=\mathscr{C}_{k_0}^{k_0+P-1}(B_k) & 
    D&:=\mathcal{T}_{k_0}^{k_0+P-1}(\uG{}{}) \\
    F&:=\mathscr{C}_{k_0}^{k_0+P-1}(F_k) &
    F&:=\mathcal{T}_{k_0}^{k_0+P-1}(\dG{}{}) \\
    K&:=\mathscr{C}_{k_0}^{k_0+P-1}(K_k) &  
    H&:=\mathcal{T}_{k_0}^{k_0+P-1}(\eG{}{}).
\end{align*}

\subsection{Applying the Internal Model Principle}\label{sec:IMP}
The internal model principle is the main mechanism behind repetitive control to mitigate disturbances. The essence of the internal model principle is that the effect of a disturbance that is generated by some signal generating model may be asymptotically attenuated by means of feedback if the controller includes the dynamics of the disturbance generating model~
\cite{Francis1976}.

In the case of a $P$-periodic disturbance $d_k=d_{k+P}$, as is the case here, this implies that the lifted disturbance $\dl_j$ is constant: $\dl_{j+1}=\dl_j$ $\forall j\in\mathbb{Z}$. Combining this disturbance model with the lifted system dynamics of \eqref{eq:lSS_inno} obtains the augmented system description
\begin{subequations}\label{eq:ldSS_inno}
\begin{empheq}[left=\mkern-8mu\mathcal{S}_\mathsf{L}^\mathrm{d}(k_0)\!:\mkern-1mu\!\empheqlbrace\!]{align}
    \begin{bmatrix}\xl_{j+1}\\ \dl_{j+1}\end{bmatrix}\!&=\!
    \left[\begin{array}{@{}c@{}c@{}}A\;\; & F \\ 0  & I_{mP}\!    \end{array}\right]
    \!\!\begin{bmatrix}\xl_j\\ \dl_j\end{bmatrix}\!+
    \!\begin{bmatrix}B\\0\end{bmatrix}\!\ul_j\!+
    \!\begin{bmatrix} K\\ 0\end{bmatrix}\!\el_j\label{eq:ldSSinno_x}\\
	\yl_j &= \begin{bmatrix} C & G \end{bmatrix}\;\!\begin{bmatrix}\xl_j\\ \dl_j\end{bmatrix} + D \ul_j + H \el_j. \label{eq:ldSSinno_y}
\end{empheq}
\end{subequations}

In data-driven control applications, minimality of the controlled system is often assumed, requiring that the system is both controllable and observable. Whilst the augmented lifted system \eqref{eq:ldSS_inno} may be observable, the modes corresponding to the dynamics of the disturbance are not controllable. The next section derives an extension of Willems' fundamental lemma for such uncontrollable systems in a (not necessarily lifted) \ac{LTI} domain that is subsequently applied in Section~\ref{sec:DeePRC_formulation} for the development of a data-driven repetitive control method that operates on lifted data from a system like~\eqref{eq:ldSS_inno}.

\subsection{Disturbances with Willems' Fundamental Lemma}\label{sec:DistWFL}
\ac{DeePC} applies Willems' fundamental lemma to make data-driven predictions in a receding horizon optimal control framework. For the purpose of the following discussion of Willems' fundamental lemma, consider a generic deterministic state-space \ac{LTI} model $\mathcal{P}$
\begin{subequations}\label{eq:SS_LTI}
\begin{empheq}[left=\mathcal{P}:\empheqlbrace]{align}
    x_{k+1} &= \mathcal{A} x_k + \mathcal{B} u_k\\
	y_k &=\mathcal{C} x_k + \mathcal{D} u_k,
  \end{empheq}
\end{subequations}
with states $x_k\in\mathbb{R}^{\nx}$, inputs $u_k\in\mathbb{R}^{n_\mathrm{u}}$, outputs ${y_k\in\mathbb{R}^{n_\mathrm{y}}}$ and system matrices $\{\mathcal{A},\mathcal{B},\mathcal{C},\mathcal{D}\}$. For the system $\mathcal{P}$ consider the following lemma.
\begin{lem}[\textbf{Willems' fundamental lemma~\cite[Th. 1]{Willems2005}}]\label{lem:WFL}
 Consider the deterministic \ac{LTI} system $\mathcal{P}$ from \eqref{eq:SS_LTI} and assume it to be controllable\footnote{Note that \cite[Th. 1]{Willems2005} employs a behavioural definition of controllability (see, e.g.,~\cite{Markovsky2021}) that is implied by classical state controllability. This latter notion of state controllability is used in~\cite{vanWaarde2020} to demonstrate the fundamental lemma for state-space representations, as well as in this work.}. Collecting input $u^\mathrm{m}_k$ and output measurements $y^\mathrm{m}_k$ during an experiment, if the input signal $\{u_k^\mathrm{m}\}_{k=0}^{N+L+\nx-2}$ is persistently exciting of order $L+\nx$, then any $L$-long input-output trajectory of $\mathcal{P}$ is described by
\begin{align}
    \begin{bmatrix}\cvec{u}{0}{L-1}\\\cvec{y}{0}{L-1}\end{bmatrix} = \begin{bmatrix}
        \blkH{0}{L}{N+\nx}{u^\mathrm{m}_k}\\
        \blkH{0}{L}{N+\nx}{y^\mathrm{m}_k}
    \end{bmatrix}g,\label{eq:WFL}
\end{align}
with $g\in\mathbb{R}^{N+\nx}$.
\end{lem}
A proof of \lemref{lem:WFL} relies on the fact that if $\mathcal{P}$ is controllable and if $u_k^\mathrm{m}$ is persistently exciting of order $L+\nx$ then \cite[Cor. 2]{Willems2005}
\begin{align}\label{eq:WFL_rank}
    \text{rank}\left(\begin{bmatrix}
        \blkH{0}{L}{N+\nx}{x^\mathrm{m}_k}\\
        \blkH{0}{L}{N+\nx}{u^\mathrm{m}_k}
    \end{bmatrix}\right)=Ln_\mathrm{u}+\nx,
\end{align}
i.e. the matrix in~\eqref{eq:WFL_rank} is is full row rank. It then follows from the Rouch\'e-Capelli theorem, that there exists a vector $g$ as above such that
\begin{align}
    \begin{bmatrix}x_0\\ \cvec{u}{0}{L-1}\end{bmatrix} = \begin{bmatrix}
        \blkH{0}{1}{N+\nx}{x^\mathrm{m}_k}\\
        \blkH{0}{L}{N+\nx}{u^\mathrm{m}_k}
    \end{bmatrix}g,\label{eq:WFL3}
\end{align}
with $x_0$ as initial state corresponding to the input-output trajectory on the left-hand side of \eqref{eq:WFL}. As is furthermore shown in~\cite[Lem. 2]{DePersis2020},~\eqref{eq:WFL} follows directly from~\eqref{eq:WFL3}.

Unlike is assumed in Willems' fundamental lemma, the lifted augmented system $\mathcal{S}_\mathsf{L}^\mathrm{d}(k_0)$ from \eqref{eq:ldSS_inno} is not controllable due to the uncontrollable disturbance modes. This raises the question of what conditions are sufficient to guarantee that \eqref{eq:WFL3}, and therefore \eqref{eq:WFL} on which \ac{DeePC} relies, still hold in the presence of a periodic disturbance.

\subsubsection{Imposing full row rank on the state-input data matrix}
Given the developments in the preceding section, a natural answer is to seek guarantees for the equivalent of~\eqref{eq:WFL_rank} using an extended state to incorporate an uncontrollable disturbance. To this end, $\mathcal{P}$ is augmented with an observable disturbance as
\begin{subequations}\label{eq:SS_LTI_dist}
\begin{empheq}[left=\mathcal{P}_\mathsf{d}:\empheqlbrace]{align}
    \xbar_{k+1} &= \underbrace{\begin{bmatrix}\mathcal{A} & \mathcal{B}_\mathrm{d}\\ 0 & \mathcal{A}_\mathrm{d}\end{bmatrix}}_{:=\mathcal{A}_\mathrm{L}} \xbar_k + \underbrace{\begin{bmatrix}\mathcal{B}\\0\end{bmatrix}}_{:=\mathcal{B}_\mathrm{L}} u_k\label{eq:SSx_LTI_dist} \\
y_k &=\begin{bmatrix}\mathcal{C} & \mathcal{C}_\mathrm{d} \end{bmatrix} \xbar_k + \mathcal{D} u_k,\label{eq:SSy_LTI_dist}
  \end{empheq}
\end{subequations}
where $\xbar_k=\left[x_k^\top\; d_k^\top\right]^\top\in\mathbb{R}^{\nxd}$ is the augmented state containing the disturbance $d_k\in\mathbb{R}^{\nd}$, with $\nxd=\nx+\nd$, and $\{\mathcal{B}_\mathrm{d},\mathcal{A}_\mathrm{d},\mathcal{C}_\mathrm{d}\}$ are extra system matrices with respect to $\mathcal{P}$ that indicate effects of the disturbance. 

What follows is a theorem that specifies sufficient conditions to guarantee that, analagous to~\eqref{eq:WFL_rank},
\begin{align}
    \text{rank}\left(\begin{bmatrix}
        \blkH{0}{1}{N+\nxd}{\xbar^\mathrm{m}_k}\\\blkH{0}{L}{N+\nxd}{u^\mathrm{m}_k}
    \end{bmatrix}\right)=Ln_\mathrm{u}+\nxd,\label{eq:Theorem1}
\end{align}
such that~\eqref{eq:WFL3} and therefore~\eqref{eq:WFL} also hold for the augmented system $\mathcal{P}_\mathsf{d}$.
\begin{thm}\label{thm:FullRowRank} Consider the deterministic \ac{LTI} system $\mathcal{P}_\mathsf{d}$ from \eqref{eq:SS_LTI_dist}. Let the superscript $\mathrm{m}$ denote data obtained from an experiment. If a measured input signal $\{u_k^\mathrm{m}\}_{k=0}^{N+L+\nxd-2}$ is persistently exciting of order $L+\nxd$, and $(\mathcal{A},\mathcal{B})$ and $(\mathcal{A}_\mathrm{d},d^\mathrm{m}_0)$ are controllable, then the resulting augmented state-input trajectories are such that~\eqref{eq:Theorem1} holds.\end{thm}
A proof of \thmref{thm:FullRowRank} follows along the same lines as~\cite[Th. 1.i]{vanWaarde2020}, which facilitates its interpretation.
\begin{proof}
 Consider a row vector $\left[\xi \;\; \eta\right]$ in the left kernel of the matrix in $\eqref{eq:Theorem1}$ with ${\xi\top\in\mathbb{R}^{\nxd}}$, and ${\eta\top\in\mathbb{R}^{n_\mathrm{u}L}}$. To prove the theorem it will be shown that 
\begin{align}\label{eq:OriginalHankel}
    \begin{bmatrix}
        \xi & \eta
    \end{bmatrix}
    \begin{bmatrix}
        \blkH{0}{1}{N+\nxd}{\xbar^\mathrm{m}_k}\\\blkH{0}{L}{N+\nxd}{u^\mathrm{m}_k}
    \end{bmatrix} = 0
\end{align} implies $\begin{bmatrix}
        \xi & \eta
    \end{bmatrix}=0$.

To start, consider that akin to \eqref{eq:OriginalHankel}, for the `tall' matrix
\begin{align}
    \begin{bmatrix}
        \xi & \eta & 0_{\nxd n_\mathrm{u}}
    \end{bmatrix}\begin{bmatrix}
        \blkH{0}{1}{N}{\xbar^\mathrm{m}_k}\\\blkH{0}{L+\nxd}{N}{u^\mathrm{m}_k}
    \end{bmatrix}=0,\label{eq:DeepHankel}
\end{align}
where $0_{\nxd n_\mathrm{u}}^\top\in\mathbb{R}^{{\nxd n_\mathrm{u}}}$. Furthermore, it is implied by~\eqref{eq:OriginalHankel} that
\begin{align}\label{eq:HankelImplied}
    \begin{bmatrix}
        \xi & \eta
    \end{bmatrix}
    \begin{bmatrix}
        \blkH{q}{1}{N}{\xbar^\mathrm{m}_k}\\
        \blkH{q}{L}{N}{u^\mathrm{m}_k}
    \end{bmatrix} &= 0,\quad \forall q\in[1,\nxd] \cap \mathbb{Z}.
\end{align}

Iterative application of~\eqref{eq:SSx_LTI_dist} yields
\begin{align}\label{eq:iterated_x}
    \xbar^\mathrm{m}_q = \Al^q \xbar^\mathrm{m}_0 + \underbrace{\begin{bmatrix}\Al^{q-1} \Bl & \Al^{q-2} \Bl & \cdots & \Bl\end{bmatrix}}_{:=\mathcal{K}_q} u^\mathrm{m}_{[0,q-1]},
\end{align}
where $\mathcal{K}_q\in\mathbb{R}^{\nxd \times qn_\mathrm{u}}$ is an extended reversed controllability matrix that is defined as shown. As a consequence of~\eqref{eq:iterated_x},
\begin{align}\label{eq:HankelIterator}
    \blkH{q}{1}{N}{\xbar^\mathrm{m}_k} = \begin{bmatrix}\Al^q & \mathcal{K}_q\end{bmatrix}
    \begin{bmatrix}
        \blkH{0}{1}{N}{\xbar^\mathrm{m}_k}\\
        \blkH{0}{q}{N}{u^\mathrm{m}_k}
    \end{bmatrix}.
\end{align}
Substituting~\eqref{eq:HankelIterator} in~\eqref{eq:HankelImplied} and adding zero columns to the row vector on the left obtains $\forall q \in [1,\nxd]$
\begin{align}\label{eq:leftKernelHankel}
    \begin{bmatrix}
        \xi\Al^q & \xi \mathcal{K}_q & \eta & 0_{(\nxd-q) n_\mathrm{u}}
    \end{bmatrix}
    \begin{bmatrix}
        \blkH{0}{1}{N}{\xbar^\mathrm{m}_k}\\
        \blkH{0}{L+\nxd}{N}{u^\mathrm{m}_k}
    \end{bmatrix}=0,
\end{align}
Together,~\eqref{eq:DeepHankel} and~\eqref{eq:leftKernelHankel} $\forall q\in[1,\nxd]$ indicate $\nxd+1$ row vectors that are in the left kernel of the tall matrix that features in both equations. These row vectors are
\begin{align*}
    w_0 &:= \begin{bmatrix}
        \xi & \eta & 0_{\nxd n_\mathrm{u}}
    \end{bmatrix}\\
    w_1 &:= \begin{bmatrix}
        \xi\Al & \xi\Bl & \eta & 0_{(\nxd-1) n_\mathrm{u}}
    \end{bmatrix}\\
    w_2 &:= \begin{bmatrix}
        \xi\Al^{2} & \xi\Al\Bl & \xi\Bl & \eta & 0_{(\nxd-2) n_\mathrm{u}}
    \end{bmatrix}\\
        \vdots&\phantom{=} \\
    w_\nxd &:= \begin{bmatrix}
        \xi\Al^{\nxd} & \xi\Al^{\nxd-1}\Bl & \cdots & \xi\Bl & \eta
    \end{bmatrix}.
\end{align*}
In addition, these $\nxd+1$ row vectors are linearly dependent because by the assumed persistency of excitation of the input of order $L+\nxd$ the dimensionality of the left kernel of the tall matrix in \eqref{eq:DeepHankel} is at most $\nxd$. This allows us to write $w_\nxd$ as a linear combination of the other vectors, from which it can iteratively be shown that all of the elements of $\eta$ are zero such that $\eta=0$.

Next, a specific linear combination of the vectors is considered. The linear combination is chosen in accordance with the Cayley-Hamilton theorem by which $\sum_{i=0}^\mathrm{\nxd}\alpha_i\Al^i=0$, with $\alpha_0=1$. Defining the new vector in the kernel of the tall matrix $v:=\sum_{i=0}^\mathrm{\nxd}\alpha_iw_i$ and substituting $\eta=0$, one finds
\begin{align*}
    v=\begin{bmatrix}
        0_\nxd & v_1 & v_2 & \cdots & v_\nxd & 0_{n_\mathrm{u}L} 
    \end{bmatrix}
\end{align*}
with $v_c=\sum_{i=c}^\mathrm{\nxd}\alpha_i\Al^{i-c}\Bl$. The nonzero part of $v$ is in the left kernel of $\blkH{0}{\nxd}{N}{u_k}$, which by the assumed persistency of excitation has a dimensionality of zero. This implies that $v_c=0$ $\forall c\in[1,2,\cdots,\nxd]$. Starting from $v_\nxd=0$ and going to $v_1=0$, by backward substitution $\xi\left[\Al^{\nxd-1}\Bl\;\;\Al^{\nxd-2}\Bl\;\;\cdots\;\;\Bl\right]=0$. Introducing the partitioning $\xi=[\xix\;\;\xid]$ with $\xix^\top\in\mathbb{R}^{\nx}$, $\xid^\top\in\mathbb{R}^{\nd}$, and expanding the matrices $\Al$ and $\Bl$ obtains
\begin{align}
    \xix
    \begin{bmatrix}
        \mathcal{A}^{\nxd-1}\mathcal{B} & \mathcal{A}^{\nxd-2}\mathcal{B} & \cdots & \mathcal{B}
    \end{bmatrix}=0.\label{eq:xix_Ctrb}
\end{align}
By the assumed controllability of $(\mathcal{A},\mathcal{B})$ the reversed extended controllability matrix in \eqref{eq:xix_Ctrb} has full row rank, which implies that $\xix=0$.

Lastly, substituting $\xix=0$ and $\eta=0$ into \eqref{eq:OriginalHankel} shows that $\xid \blkH{0}{1}{N+\nxd}{d_k^\mathrm{m}}=0$. With reference to \eqref{eq:SSx_LTI_dist} this may be rewritten as
\begin{align}
    \xid \begin{bmatrix}
        d_0^\mathrm{m} & \mathcal{A}_\mathrm{d} d_0^\mathrm{m}& \cdots & \mathcal{A}_\mathrm{d}^{N+\nxd-1}d_0^\mathrm{m}
    \end{bmatrix}=0.\label{eq:xid_Ctrb}
\end{align}
Similarly as before, the controllability of $(\mathcal{A}_\mathrm{d},d^\mathrm{m}_0)$ ensures that \eqref{eq:xid_Ctrb} implies that $\xid=0$.
\end{proof}

The above \thmref{thm:FullRowRank} provides sufficient conditions for the rank condition stipulated by~\eqref{eq:Theorem1} to hold. As explained in the beginning of this section, this rank condition is an important building block of typical \ac{DeePC} implementations. Unfortunately, the conditions posed by \thmref{thm:FullRowRank} are rather restrictive. To see this, compare $\mathcal{P}_\mathsf{d}$ from \eqref{eq:SS_LTI_dist} to the non-stochastic component of $\mathcal{S}_\mathsf{L}^\mathrm{d}(k_0)$ from \eqref{eq:ldSS_inno}. In the lifted framework the periodic disturbance is constant such that $\Al=I_{mP}$. With reference to \thmref{thm:FullRowRank}, this implies that the lifted disturbance $\dl_j\in\mathbb{R}^{mP}$ must be full row rank, which is highly unlikely given that this would require that $m=1$ and that $P=1$. Particularly this latter condition would obviate the need for lifting in the first place since the considered \ac{LPTV} system is then actually \ac{LTI}.

\subsubsection{Relaxing the assumption of controllability}
Since the controllability of $(\mathcal{A}_\mathrm{d},d^\mathrm{m}_0)$ seems like a difficult condition to satisfy for many kinds of exogenous disturbance generators, it makes sense to look into whether this condition, and consequently~\eqref{eq:Theorem1}, may be relaxed. As it turns out, the conditions presented by \thmref{thm:FullRowRank} are sufficient but not necessary to ensure an equivalent of \eqref{eq:WFL3} upon which \ac{DeePC} ultimately relies. The following theorem formalizes this insight.
\begin{thm}[\textbf{Fundamental lemma for systems with an exogenous disturbance}]\label{thm:WFLv2}
For the \ac{LTI} system $\mathcal{P}_\mathsf{d}$ from \eqref{eq:SS_LTI_dist}, if the pair $(\mathcal{A},\mathcal{B})$ is controllable, the controllability matrix of $(\mathcal{A}_\mathrm{d},d^\mathrm{m}_0)$ has rank $\nu$ with $1\leq\nu<\nd$, and a measured input signal $\{u_k^\mathrm{m}\}_{k=0}^{N+L+\nxd-2}$ is persistently exciting of order $L+\nxd$, then the following three properties hold:
\begin{enumerate}
    \item[(i)] the matrix from~\eqref{eq:Theorem1} is $\nd-\nu$ rank deficient:
    \begin{align}\label{eq:WFLv2i}
        \textnormal{rank}\left(\begin{bmatrix}
        \blkH{0}{1}{N+\nxd}{\xbar^\mathrm{m}_k}\\\blkH{0}{L}{N+\nxd}{u^\mathrm{m}_k}
    \end{bmatrix}\right)=Ln_\mathrm{u}+\nx+\nu,
    \end{align}
    \item[(ii)] $\exists g$ such that
    \begin{align}\label{eq:WFLv2ii}
        \begin{bmatrix}
            \blkH{0}{1}{N+\nxd}{\xbar^\mathrm{m}_k}\\
            \blkH{0}{L}{N+\nxd}{u^\mathrm{m}_k}
        \end{bmatrix}g=\begin{bmatrix}\xbar_0\\ \cvec{u}{0}{L-1}\end{bmatrix},
    \end{align}
    \item[(iii)] $\exists g$ such that for measured inputs $u_k^\mathrm{m}$ and outputs $y_k^\mathrm{m}$,
    \begin{align}\label{eq:WFLv2iii}
        \begin{bmatrix}
        \blkH{0}{L}{N+\nxd}{u^\mathrm{m}_k}\\
        \blkH{0}{L}{N+\nxd}{y^\mathrm{m}_k}
    \end{bmatrix}g=
    \begin{bmatrix}\cvec{u}{0}{L-1}\\
    \cvec{y}{0}{L-1}\end{bmatrix}.
    \end{align}
\end{enumerate}\end{thm}
\begin{proof}
Proof of $(i)$ follows along the lines of the proof of \thmref{thm:FullRowRank}: $\xix = 0$, $\eta=0$, but since the controllability matrix of $(\mathcal{A}_\mathrm{d},d^\mathrm{m}_0)$ is $\nd-\nu$ rank deficient, so are $\blkH{0}{1}{N+\nxd}{d_k^\mathrm{m}}$ and the matrix in~\eqref{eq:WFLv2i}.

For the proof of $(ii)$, consider the following kernel representation of~\eqref{eq:SSx_LTI_dist}
\begin{align}\label{eq:LeftNull}
    \underbrace{\begin{bmatrix}
        -I_\nx & 0 & \mathcal{A} & \mathcal{B}_\mathrm{d} & \mathcal{B} & 0\\
        0 & -I_\nd & 0 & \mathcal{A}_\mathrm{d} & 0 & 0\\
        0 & 0 & 0 & \Xi_\mathrm{d} & 0 & 0\\
    \end{bmatrix}}
    \begin{bmatrix}
        x_{k+1}\\ d_{k+1} \\ x_k \\ d_k \\ \cvec{u}{k}{k+L-1}
    \end{bmatrix} = 0,
\end{align}
where the rows of ${\Xi_\mathrm{d}\in\mathbb{R}^{(\nd-\nu)\times\nd}}$ form a minimal basis that spans the left kernel of the controllability matrix of $(\mathcal{A}_\mathrm{d},d_0^\mathrm{m})$. 
By inspection of~\eqref{eq:LeftNull} it is possible to find a minimal basis of the right nullspace of the underbraced matrix, for which there then exists a linear combination $g$ such that
\begin{align}\label{eq:gNull}
    \left[\begin{array}{@{}cccc@{}}
         \mathcal{A} & \mathcal{B}_{\mathrm{d}}\mathsf{C}_\mathrm{d} & \mathcal{B} & 0\\ 
        0 & \mathcal{A}_\mathsf{d}\mathsf{C}_\mathrm{d} & 0 & 0\\ \hline
        I_{\nx} & 0 & 0 & 0\\
        0 & \mathsf{C}_\mathrm{d} & 0 & 0\\
        0 & 0 & I_{n_\mathrm{u}} & 0 \\
        0 & 0 & 0 & I_{n_\mathrm{u}(L-1)} 
    \end{array}\right]
    g=\left[\begin{array}{@{}c@{}}
        x_{k+1}\\ d_{k+1} \\ \hline x_k \\ d_k \\ \cvec{u}{k}{k+L-1}
    \end{array}\right],
\end{align}
where $\mathsf{C}_\mathrm{d}\in\mathbb{R}^{\nd\times \nu}$ is a matrix whose columns form a minimal basis of the columnspace of the controllability matrix of $(\mathcal{A}_\mathrm{d},d_0^\mathrm{m})$. By statement $(i)$ the columns of the data matrix on the left of~\eqref{eq:WFLv2ii} must span the columnspace of the bottom part of the matrix in \eqref{eq:gNull}, which proves $(ii)$.

Proof of statement $(iii)$ follows directly from $(ii)$, see \cite[Lem.~2]{DePersis2020}.
\end{proof}
The above theorem demonstrates that the controllability and rank condition posed by Willems' fundamental lemma can be relaxed to facilitate the application of \ac{DeePC} to systems of the form given by~\eqref{eq:SS_LTI_dist}. 

\subsubsection{Implications for a constant disturbance}
The following corollary states the implications of \thmref{thm:WFLv2} for a constant disturbance, as is the case in the lifted domain of
~\eqref{eq:ldSS_inno}.
\begin{cor}\label{cor:WFLdconst} Let the conditions of \thmref{thm:WFLv2} hold for a nonzero constant disturbance such that $\mathcal{A}_\mathrm{d}=I_\nd$,  then the results of \thmref{thm:WFLv2} hold with $\nu=1$.
\end{cor}

\ac{LTI} systems with a constant disturbance are a type of affine system, for which an alternative proof of \corref{cor:WFLdconst} can be found in the combination of~\cite[Th.~1]{Martinelli2022} and~\cite[Th.~1]{Berberich2022}. The above corollary enables direct data-driven control of the lifted system on the basis of~\eqref{eq:WFLv2iii} using \ac{DeePC} if the predicted output is additionally unique. This is the case when the past window length that is used to approximate the initial state is larger than the lag of the lifted system~\cite[Lem. 7.ii]{Li2022}. The next section presents the control problem formulation and noise mitigation strategy.

\subsection{DeePRC Formulation and Noise Mitigation}\label{sec:DeePRC_formulation}
Having lifted the \ac{LPTV} system $\mathcal{S}$ of~\eqref{eq:SS_inno} to the \ac{LTI} form $\mathcal{S}_\mathsf{L}^\mathrm{d}(k_0)$ of~\eqref{eq:ldSS_inno}, and having shown how \ac{DeePC} may accommodate exogenous disturbances in such an \ac{LTI} domain in the previous section, this section develops a data-driven repetitive control method that operates on lifted data to accommodate periodic behaviour and disturbances and mitigates noise.

\subsubsection{Closed-loop DeePC applied to a lifted system}
Use will be made of the computationally efficient \ac{CL-DeePC} framework developed in~\cite[Sec.~4]{Dinkla2024} for several reasons. Firstly, the method uses the available data relatively efficiently to suppress noise when compared to \ac{DeePC}. Secondly, the dimension of the identification task is considerably reduced w.r.t. \ac{DeePC}, which is especially significant for lifted systems. Thirdly, we note the potential of this method to obtain consistent output predictions from noisy closed-loop data,\footnote{For lifted systems this necessitates using a suitably chosen instrumental variable $\mathcal{Z}$. An optimal choice is non-trivial, so for simplicity~\eqref{eq:CL-DeePC} uses a common choice for open-loop data that is sub-optimal otherwise. For dedicated alternatives for closed-loop data see \cite{Wang2023b}.} which is a problem for regular~\ac{DeePC}~\cite{Dinkla2023}. For clarity we provide an analytically equivalent, simpler representation of the employed \ac{CL-DeePC} framework from~\cite[Sec.~4]{Dinkla2024} as

\begin{subequations}\label{eq:CL-DeePC}
\begin{empheq}[]{align}
    &\underbrace{\left[\begin{array}{@{}c@{}}
        \blkH{i}{p}{N}{\ul_j}\\\blkH{i_p}{1}{N}{\ul_j}\\\blkH{i}{p}{N}{\yl_j}
    \end{array}\right]}_{:=\mathcal{Z}}
    \mathcal{Z}^\top\!G
    \!=\!\begin{bmatrix}
        \blkH{\hat{i}}{p}{f}{\ul_j}\\\blkH{\hat{i}_p}{1}{f}{\ul_j}\\\blkH{\hat{i}}{p}{f}{\bar{\yl}_j}
    \end{bmatrix},\\
    &\blkH{i_p}{1}{N}{\yl_j}\mathcal{Z}^\top \underbrace{\begin{bmatrix}g_1 & \!\!\!\cdots & \!\!\!g_f\end{bmatrix}}_{:=G}=\blkH{\hat{i}_p}{1}{f}{\hat{\yl}_j},
  \end{empheq}
\end{subequations}
where $i$ and $\hat{i}$ are starting indices of the relevant data matrices, $p$ is a window length of past data, $f$ is the prediction window length, $i_p:=i+p$, $\hat{i}_p:=\hat{i}+p$, $\hat{(\cdot)}$ indicates predictions, and $\bar{(\cdot)}$ indicates that the data is composed in part of predictions. With this notation, the first future sample is found at $\hat{i}_p$. Furthermore, $G\in\mathbb{R}^{((p+1)Pr+pPl)\times f}$ is a collection of $f$ vectors that are found in \ac{DeePC} with instrumental variables~\cite{vanWingerden2022}. This paper uses $\mathcal{Z}$ as an instrumental variable to mitigate the effects of noise whilst preserving the rank of the matrix pre-multiplying $G$. Subsequent columns of~\eqref{eq:CL-DeePC} correspond to \ac{DeePC} with an instrumental variable matrix and a prediction window length of one, applied to the same past data matrix to find trajectories that are each subsequently shifted one (lifted) sample into the future.

\subsubsection{Optimal Control Problem Formulation} The optimal control problem formulation that is used here is
\begin{subequations}\label{eq:OCP}
\begin{align}
    \min_{\cvec{\ul}{\hat{i}_p}{\hat{i}_p+f-1}}
    &||\cvec{\hat{\yl}}{\hat{i}_p}{\hat{i}_p+f-1}||_Q^2 + ||\cvec{\ul}{\hat{i}_p}{\hat{i}_p+f-1}||_R^2\label{eq:OCP_cost}\\
\text{s.t. }&\text{\eqref{eq:CL-DeePC}, }\\
&\ul_j\in\mathcal{U},\;\hat{\yl}_j\in\mathcal{Y},\; \forall j\in[\hat{i}_p,\;\hat{i}_p+f-1],
\end{align}
\end{subequations}
in which, $R$ and $Q$ are respectively positive(semi)-definite weighting matrices, and $\mathcal{U}$ and $\mathcal{Y}$ are allowed sets constraining the inputs and outputs.

\subsubsection{Receding horizon implementation} 
With regards to the receding horizon implementation, consider implementing either the entire sequence of inputs corresponding to the first computed lifted input $\ul_{\hat{i}_p}$, or the first input of this sequence $u_{k_0+\hat{i}_pP}$. A significant advantage of the latter method is the availability of more frequent feedback by which to improve performance. However, note that for data-driven control applications that seek to learn \ac{LTI} behaviour (e.g. by updating RQ data factorizations at each time step as in~\cite{Hallouzi2010}) the latter method also has an important drawback: since the lifted \ac{LTI} system varies periodically with the starting point ${\mathcal{S}_\mathsf{L}^\mathrm{d}(k_0)=\mathcal{S}_\mathsf{L}^\mathrm{d}(k_0+P)}$, such applications would need to learn the behaviour of $P$ different lifted \ac{LTI} systems.
\section{Simulation Results}\label{sec:Simulation}
Having presented the \ac{DeePRC} framework, this section will demonstrate the superior performance it can achieve for a periodic system with a periodic disturbance.

\subsection{Simulated LPTV System}
The simulated periodic system was obtained from~\cite{vanWingerden2008} and constitutes a linear parameter-varying system with a periodic scheduling parameter $\mu_k=\cos(\tfrac{2\pi}{P}k)$, with period $P=20$. The periodic matrices from \eqref{eq:SS_inno} are defined by
{\footnotesize%
\begin{align*}
    \left[\!\!\begin{array}{c|c}
        A^{(1)} & A^{(2)}
    \end{array}\!\!\right] &= \left[\begin{array}{ccc|ccc}
           0 & 0.9 & 0.2 & 0.6 & 0.5 & 0.5 \\
        -0.9 & 0.5 & 0 & 0.5 & 0.6 & 0 \\
        -0.2 & 0 & 0.2 & -0.5 & 0 & 0.6
    \end{array}\right], \\
    \left[\!\!\begin{array}{c|c}
        B^{(1)} & B^{(2)}
    \end{array}\!\!\right] &= \left[\begin{array}{c|c}
        1 & 0.4 \\
        1 & 0.2 \\
        1 & 0.12
    \end{array}\right]\!\!,\;
    \left[\!\!\begin{array}{c|c}
        D^{(1)} & D^{(2)}
    \end{array}\!\!\right] = \left[\begin{array}{c|c}
        0.1 & 0.2 \\
        0.2 & 0.1
    \end{array}\right]\!, \\
    \left[\!\!\begin{array}{c|c}
        C^{(1)} & C^{(2)}
    \end{array}\!\!\right] &= \left[\begin{array}{ccc|ccc}
        0.2 & 1 & 0.5 & 0.2 & 0.1 & 1 \\
        0.2 & 0.1 & 1 & 0.3 & 0.4 & 0.8
    \end{array}\right],\\
    \left[\!\!\begin{array}{c|c}
        K^{(1)} & K^{(2)}
    \end{array}\!\!\right] &= \left[\begin{array}{cc|cc}
        0.0130 & 0.0225 & 0 & 0 \\
        0.0089 & 0.0060 & 0 & 0 \\
        0.0002 & -0.0010 & 0 & 0
    \end{array}\right],
\end{align*}}%
with $A_k = A^{(1)}+\mu_k A^{(2)}$, and likewise for $B_k$, $C_k$, $D_k$, and $K_k$. An unknown input disturbance will be applied such that $F_k=B_k$, and $G_k=D_k$. The periodic disturbance is given by $d_k=\sin(\tfrac{2\pi}{P}k)$. Use is made of zero-mean white innovation noise $e_k$ with a variance of $0.05$. In the lifted domain represented by \eqref{eq:ldSS_inno} this system has controllable $(A,B)$ and is observable.

\subsection{Controller Settings}
Two different controllers are simulated. One controller makes use of the \ac{DeePRC} framework with optimal control problem formulation given by \eqref{eq:OCP} with $p=1$, $f=2$ periods, $Q=100$, $R=1$. This controller computes a lifted future input sequence and implements the first sample ${u_k\in\mathbb{R}^r}$ thereof. A second controller uses \ac{CL-DeePC} and operates fully in the non-lifted domain, likewise implementing the first sample $u_k$ that it computes. The optimal control problem solved by the second controller is comparable to that solved by \ac{DeePRC}. Furthermore $Q$ and $R$ are the same for the second controller, as are the effective window lengths in terms of the non-lifted domain. Both controllers implement constraints of the form $|u_k|\leq 10$ and $|y_k|\leq 20$ and are initialized with $1000$ periods of open-loop data where $u_k$ is zero-mean white noise with a variance of $1$.

\subsection{Simulation Results}
With the above simulation model and controller settings, the obtained results are shown in Fig. \ref{fig:sim_results1}. The controllers are enabled at the grey vertical line when the open-loop data collection ends. The figure clearly demonstrates that the \ac{DeePRC} controller, which operates in the lifted domain, outperforms the \ac{CL-DeePC} controller in attenuating the effect of the disturbance to regulate the output channels to zero. This is because the \ac{DeePRC} controller effectively operates on \ac{LTI} data and is therefore better able to form an implicit internal disturbance model, as explained in Section~\ref{sec:DistWFL}. Furthermore, notice that after initialization the \ac{CL-DeePC} controller violates the output constraints. This is possible because the constraints are formulated only for future samples and are, if necessary to ensure feasibility, relaxed. Moreover, there is considerable mismatch for the \ac{CL-DeePC} controller between the data-driven output predictor that it employs and the true system. This contributes to a deterioration of the performance compared also to the case where no control action is applied.
\begin{figure}[t!]
      \centering
      \includegraphics[width=\columnwidth]{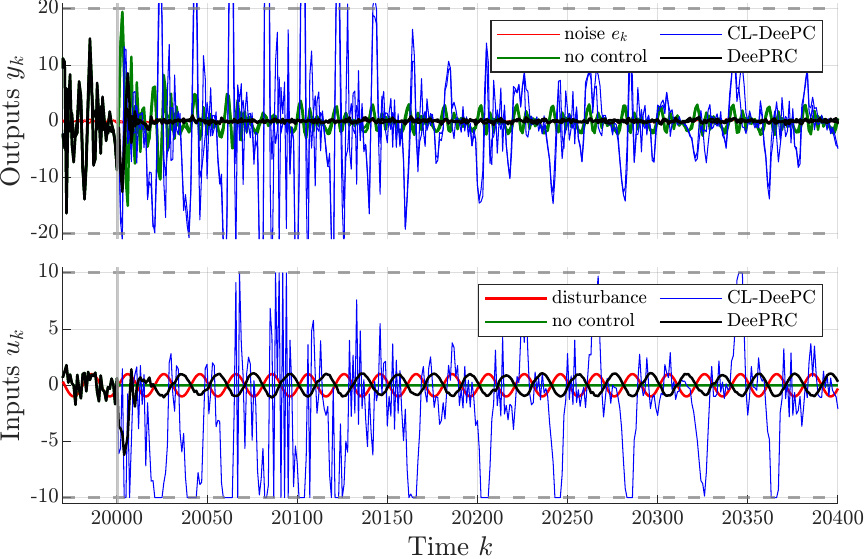}
      \caption{Performance of \ac{DeePRC} and \ac{CL-DeePC} controllers using respectively data from the lifted \ac{LTI} and the non-lifted periodic domain under the influence of a periodic disturbance and noise. \ac{DeePRC} effectively compensates the input disturbance, whilst \ac{CL-DeePC} performs worse here w.r.t. the case without control. Dashed grey lines indicate contraints.}
      \label{fig:sim_results1}
\end{figure}

The obtained iteration cost, as specified by
\begin{align}
    \mathcal{J}(j)=||\yl_j||_Q^2 + ||\ul_j||_R^2\label{eq:iter_cost}    
\end{align}
can be calculated for each iteration index $j$. The results are shown in Fig.~\ref{fig:sim_results2} for both the setting with noise shown in Fig.~\ref{fig:sim_results1} and in the same setting, but without noise. In the noiseless case, the \ac{DeePRC} iteration cost demonstrates convergence that is, at least initially, quite fast. Moreover, the iteration cost of \ac{DeePRC} is considerably lower than would be the case without control. This is not the case for the \ac{CL-DeePC} controller, which only appears to achieve a somewhat better iteration cost than no control would in the absence of noise and after a considerable number of iterations.
\begin{figure}[t!]
      \centering
      \includegraphics[width=\columnwidth]{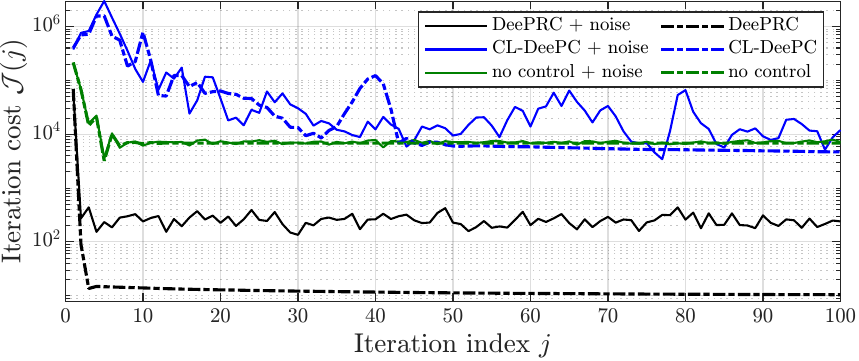}
      \caption{Obtained iteration cost as specified by~\eqref{eq:iter_cost} of the \ac{DeePRC} and \ac{CL-DeePC} controllers under conditions with and without noise. The iteration cost of \ac{DeePRC} is lower than that of \ac{CL-DeePC} and illustrates faster convergence.}
      \label{fig:sim_results2}
\end{figure}
\section{Conclusions and Future Work}\label{sec:Conclusion}
A new control framework is presented that is able to address both dynamics and disturbances of a known period in the presence of noise. Moreover, it is shown under what conditions Willems' fundamental lemma can accommodate autonomous, uncontrollable disturbance dynamics that arise from the application of the internal model principle. In particular, it is shown how despite a loss of the generally assumed controllability, a constant disturbance may still be accommodated by \ac{DeePC} formulations. Simulation results indicate superior performance of the \ac{DeePRC} controller compared to a \ac{CL-DeePC} controller that respectively use data from the lifted \ac{LTI} and periodic system domains. Future work considers the effect of an uncertain period as well as periodic data differencing to remove the effect of the disturbance from the data.
\addtolength{\textheight}{-0.75cm}   
\bibliographystyle{IEEEtran}
\bibliography{main}
\end{document}